\documentclass[11pt]{llncs}
\usepackage{fullpage}

\usepackage{amsmath, amssymb}
\usepackage{latexsym}
\usepackage{graphicx}
\usepackage{ifthen}
\usepackage{subcaption}
\usepackage{algorithm}
\usepackage[noend]{algpseudocode}
\usepackage[sectionbib,numbers]{natbib}
\usepackage{mathtools}
\usepackage{dsfont}
\usepackage{nicefrac}
\usepackage{wrapfig}
\usepackage{mfirstuc}
\usepackage{tikz}
\usetikzlibrary{arrows.meta,chains,decorations.pathreplacing}
\usetikzlibrary{calc, matrix, positioning}

\usepackage{hyperref}
\hypersetup{colorlinks=true,linkcolor=blue,citecolor=blue,urlcolor=blue}

\usepackage{thm-restate}

\newcommand{\stam}[1]{}

\newcommand{\event}[2]{{G^{#2}_{#1}}}

\newcommand{\topG}[1]{T_{#1}}

\newcommand{\menrank}{R_M}
\newcommand{\womenrank}{R_W}
\newcommand{\slowfunc}{{\tau}}
\renewcommand{\O}{\mathcal{O}}

\begin{document}
	\title{The Influence of One Strategic Agent on the Core of Stable Matchings}
	\author{%
		Ron Kupfer\thanks{This project has received funding from the European Research Council (ERC) under the European Union's Horizon 2020 research and innovation program (grant agreement No. 740282).
		} \\
	}
\authorrunning{Ron Kupfer}
\institute{The Hebrew University of Jerusalem, Israel\\
\email{ron.kupfer@mail.huji.ac.il}}
	\date{}	
	\maketitle	
	\begin{abstract} 
		In this work, we analyze the influence of a single strategic agent on the quality of the other agents' matchings in a matching market. We consider a stable matching problem with $n$ men and $n$ women when preferences are drawn uniformly from the possible $(n!)^{2n}$ full ranking options. We focus on the effect of a single woman who reports a modified preferences list in a way that is optimal from her perspective.
		We show that in this case, the quality of the matching dramatically improves from the other women's perspective. When running the Gale--Shapley men-proposing algorithm, the expected women-rank is $O(\log^4 n)$ and almost surely the average women-rank is $O(\log^{2+\epsilon}n)$, rather than a rank of $O(\frac{n}{\log n})$ in both cases under a truthful regime. On the other hand, almost surely, the average men's rank is no better than $\Omega\left(\frac{n}{\log^{2+\epsilon}n}\right)$, compared to a rank of $O({\log n})$ under a truthful regime.\\
		All of the results hold for any matching algorithm that guarantees a stable matching, which suggests that the core convergence observed in real markets may be caused by the strategic behavior of the participants.
	\end{abstract}
		
	\section{Introduction}
The stable matching problem concerns a scenario where we must find a matching between two disjoint sets of agents that satisfies natural stability constraints. This problem has received an enormous amount of attention, starting with the seminal work of \citet{gale1962college}, and has been used as a paradigm in a host of many applications.

The basic formalism considers a matching between a set of $n$ men and $m$ women. Each man has a preference order over all women and the option of staying unmatched,
and the same goes for the women. A matching between the set of men and the set of women is called \textit{stable} if there exists no ``blocking pair'', i.e. a man and a woman who prefer each other to their current matching.  

The Men-Proposing Deferred-Acceptance algorithm (Gale--Shapley algorithm, DA) is an algorithm for finding a stable matching and its proof of correctness shows that such a matching always exists.
The algorithm works iteratively: at each round, every unmatched man proposes to his most preferred woman who has not previously rejected him. Each woman chooses her most preferred man out of those who proposed to her and releases all the other proposers to continue on their lists. The algorithm terminates when all men are matched or reached to the end of their list.

It is known that the algorithm is optimal for men in the following sense: for each man $m$, there is no other stable matching in which $m$ is matched with someone whom he prefers more. On the other hand, as shown by \citet{mcvitie1971stable}, the algorithm yields the worst stable matching for any woman. That is, for each woman $w$, there is no other stable matching in which $w$ is matched with someone whom she prefers less.

Starting with the seminal contribution of \cite{gale1962college}, and for several decades that followed, economists believed that one of the most crucial market design decisions to be made by policymakers in the process of implementing the Deferred Acceptance algorithm is to choose the proposing side. Indeed, on top of the men-optimality of the men-proposing DA, it was readily shown by \citet{knuth1976stable} that in a likely market, the difference in expected rank of partners between the two sides is anything but marginal: while men get on average their $\log(n)$-th most preferred woman, women only get their $n/\log(n)$-th most preferred man.\footnote{Unless stated otherwise, we use $\log(n)$ to denote the natural logarithm.}

This idea, that proposing gives a serious advantage, remained prominent among researchers until very recently. The first hint that something else is going on came from \citet{roth1999redesign} work on the NRMP data. They showed that in fact the difference between the doctor-proposing DA and the program-proposing DA is not that significant, and only a handful number of doctors and programs get different allocations under these two mechanisms. It took nearly twenty years more, and several important theoretical insights, for market designers to fully understand that the classic matching theory result depended on so many factors and that in most real markets it is very unlikely to see any difference between the two extremal stable allocations. Following simulations presented by \cite{roth1999redesign}, \citet{immorlica2005marriage} proved that if agents on one of the sides consistently rank only a small subset of the agents on the other side, then this side gets a significant advantage in any stable mechanism. This intuition was later extended by \citet{kojima2009incentives} to the case of many-to-one matching as well. More recently, \citet{ashlagi2017unbalanced} showed that even the slightest imbalance in the number of men and women in the market causes the short side of the market to gain the advantage. Generally speaking, the advantage lies with the side that has some kind of ``market power.''\footnote{In real data, the market power often lies in the hands of few agents, and not the entire side, so all results should be interpreted as pointing to theoretical reasoning and logical paths through which certain agents get a better allocation.}

This paper continues this line of work by presenting a new dimension which is likely to affect real markets and may also shift the balance of power. In order to clarify this new intuition, we present a model that intentionally avoids all the previous arguments for core-convergence (i.e., similarity within the set of stable matchings). Specifically, our agents hold complete rank order lists, and the market is balanced. The result we prove is strikingly powerful: We show that even the slightest introduction of strategic behavior already gives the advantage to the proposed-to side. This immediately suggests that the core convergence we observe in field data may also be caused by the presence of even a small number of strategic agents.

Under the Gale--Shapley mechanism it is a dominant strategy for men to report their preferences truthfully \cite{dubins1981machiavelli}. This is not the case for women.
The following simple example with two men ($m_1$,$m_2$) and two women ($w_1$,$w_2$) demonstrate this. 
Suppose that their preference rankings are as follows:
$$\vcenter{\halign{#\hfil\quad&$#$\hfil\qquad\qquad
		&#\hfil\quad&$#$\hfil\cr
		$m_1$ prefers&w_1>w_2&$w_1$ prefers &m_2>m_1\cr
		$m_2$ prefers&w_2>w_1&$w_2$ prefers &m_1>m_2\ ,\cr
}}$$
and for all of the agents the least preferred option is to stay single. The algorithm matches $w_1$ to $m_1$. If $w_1$ falsely reports that she prefers staying single to being matched with $m_1$, the algorithm matches her to $m_2$, whom she prefers to $m_1$.

Strategic behavior under stable matching algorithms has been a topic of vast research.
In \cite{roth1982economics} it is shown that there is no algorithm for which reporting the true preferences is a dominant strategy for both men and women.
A partial list of works on strategic behavior by women under the Gale--Shapley algorithm includes \cite{dubins1981machiavelli,gale1985ms,immorlica2005marriage,coles2014optimal,teo2001gale,roth1999truncation,ma2010singleton}.

Notice that strategic behavior by a woman affects the outcome of the other women as well, e.g., in the example above, perhaps surprisingly, the strategic behavior of $w_1$ improves the matching of $w_2$. It is known \cite{gonczarowski2013sisterhood,ashlagi2012manipulability} that in general, under the Gale--Shapley algorithm, strategic behavior by any set of women can only benefit the other women in the sense that if none of the former are worse off, then neither are any of the latter.
A natural and well-known way to manipulate the algorithm is to truncate the list of the reported preferences, i.e., to set a threshold so that only mates from some given rank and above are acceptable (e.g., \cite{gonczarowski2014manipulation}).
This strategy is optimal for any woman when all the other agents' preferences are known to her, assuming that all the other agents report truthfully \cite{roth1999truncation}. In the same work, it is shown that this is also true for a wide range of partial information structures.

In this work, we analyze the effect of a selfish reporting strategy by one or more women, on the quality of the matching obtained by the other agents. 
Specifically, we explore the characteristics of such strategies in the commonly studied setting of a balanced market with $n$ men and $n$ women (e.g., \cite{irving1986complexity,irving1987efficient,vate1989linear}) and the set of preferences is drawn uniformly at random from the set of all possible $(n!)^{2n}$ full ranking options (e.g., \cite{pittel1989average,knuth1990stable,coles2014optimal,pittel2007number}). 
A simple measure of the quality of a match from the perspective of a given agent is the rank of her or his match (e.g., \cite{irving1987efficient}). We say that a person has a rank $k$ in a matching if they are matched with their $k$th favorite mate (where rank $1$ denotes being matched to the most preferred mate and rank $n$ denotes being matched to the least preferred one). 
As stated before, in our model, under a truthful regime, the expected rank of any woman $w$ is of order $n/{\log(n)}$ and The expected rank of the men is of order $\log(n)$. 

Our main theorem shows that one strategic agent is expected to dramatically affect the entire outcome of the matching. In particular, the average women's rank is polylogarithmic in $n$ compared to an expected average rank of order of $n/{\log(n)}$ under a truthful reporting regime. These results hold in expectation and with high probability.
Formally:
\begin{theorem} \label{thm: w.h.p wlog2}
	In a random uniformly ranked balanced matching market with $n$ men and $n$ women, where a single woman uses her optimal strategy and all the other agents report truthfully, we have that:
	\begin{enumerate}
		\item Almost surely, a $1-o(1)$ fraction of the women are matched with a man from their top $\O\left(\log^{2+\epsilon}(n)\right)$ men.
		Furthermore, the expected average women's rank is of order of $\O\left(\log^4(n)\right)$.
		\item Almost surely, the average men's rank is no better than $\Omega\left(\frac{n}{\log^{2+\epsilon}(n)}\right)$. 
	\end{enumerate}
\end{theorem}

These results do not hinge on the strategic woman having full information of the preferences or taking the exact optimal truncation but are in fact more robust, as we will show later.

Since the Gale--Shapley algorithm yields the worst stable matching for women, any upper bound for the women's rank holds for any matching algorithm that outputs a stable matching for the true preferences. 
Assume that the strategic woman considers only truncation strategies. The optimal truncation strategy is indeed an optimal strategy for her among all the possible strategies. Since an optimal truncation strategy yields a stable matching in the true preferences under any matching algorithm that guarantees to output a stable matching for the reported preferences, our results hold for a large family of algorithms.
Such algorithms are discussed in ,e.g., \cite{roth1990random,dworczak2016deferred,klaus2006procedurally,aldershof1999refined,irving1987efficient,gusfield1987three}.

We present simulation results in the setting described above, for different market sizes. These simulation results reinforce the theorem's conclusions and suggest that the actual effect might be even stronger than the formally proven effect.

One interpretation of the results is that although the Gale--Shapley algorithm is optimal for the men, in fact, strategic behavior by only one woman is sufficient to almost completely eliminate the men's advantage.
An important insight is that strategic behavior narrows the set of possible stable matchings attainable by any algorithm. That is, one strategic agent may strongly influence the matching in favor of that agent's sex, be it the men or the women. Our result leads to the conclusion that even a small amount of strategic behavior is sufficient in order to rule out almost all of the possible matching that are stable under the true players' preferences.

\section{Preliminaries}
\paragraph*{The Process.}
We analyze the effect of one strategic agent in a balanced market of $n$ men and $n$ women where the set of preferences are chosen uniformly from the possible $(n!)^{2n}$ full ranking options. 
In this section, we assume that one woman, $g$, gets to look at all the other agents' preference lists in advance. Then, $g$ acts according to a preference list chosen by her in a strategic way in order to maximizes her utility.

The outcome of the Gale--Shapley algorithm is independent of the order of proposals made by the men \cite{mcvitie1971stable}.
Using this fact, the algorithm can be described as the following process: a man needs to choose his $i$th preference only after the first $i-1$ women have already rejected him.
The fact that the preferences are chosen uniformly allows another simplification. When a man chooses his $i$th woman, he draws a preference with ``amnesia''; i.e., the realization is done over all $n$ possibilities, allowing him to choose a woman who has already rejected him and make her a redundant proposal. In addition, each woman also reveals her preferences in an online manner and she accepts her $k$th (distinct) proposal with probability $\frac{1}{k}$. This point of view was first presented in \cite{knuth1990stable}.\\ 
Under such a perspective, the running of the algorithm under $g$'s strategy can be viewed as a process in which women getting proposals and $g$ keeps rejecting all of the offers she gets. At some point, $g$ decides to accept an offer and the process terminates. This stopping point is selected optimally by $g$ given the preferences of all the agents. As we will show, the outcome is a stable matching in the original preferences.
Combining these observations, we follow \cite{knuth1990stable} and describing the process using the following algorithm.
In order to account for the property that $g$'s strategy is optimal, we add an oracle that tells $g$ when to stop the rejection process. That is, the oracle is exposed to the realization of all the players' preferences yet to come. This process is equivalent to the original Gale--Shapley algorithm with $g$ acting in a strategically optimal way. We add some notations to be used in the algorithm description:
\\For $i\in[n]$, we denote by $W_i$ the set of women proposed to by man $i$.
For $i\in[n]$, we denote by $x_i$ the man who have made the best offer so far to woman $i$ and by $k_i$ the number of proposals received by this woman.
$\ell$ is the number of men who have proposed at least once so far.
$p$ is the man who is currently proposing and $h$ is the woman who is currently being proposed to.
We now describe the process: 
\begin{enumerate}
	\item  Let $W_j = \emptyset$, $x_j = 0$ and $k_j = 0$ for $1 \leq j \leq n$; also let $\ell = 0$.
	\item  If $\ell < n$, increase $\ell$ by $1$ and let $p = \ell$. Otherwise, if $g$'s oracle tells her to stop, the process terminates and this is the final matching. Otherwise, let $g$ reject $x_g$ and set $p = x_g$.
	\item  Let $h$ be a random number uniformly chosen between $1$ and $n$. We say that man $p$ proposes to woman $h$. If $h \in W_p$ (i.e., if $p$'s proposal is redundant), repeat this step. Otherwise replace $W_p$ with $W_p \cup \{h\}$ and go on to step 4.
	\item  Increase $k_h$ by one. With probability $1 - \frac{1}{k_h}$, return to step 3 (in this case we say that woman $h$ rejects the proposal). Otherwise interchange $p \leftrightarrow x_h$ (that is, $h$ accepts the proposal and her former match is the next proposer). If the new value of $p$ is zero, or if $h = g$, go back to step 2;
	otherwise continue with step 3.
\end{enumerate}
We denote by $\womenrank$ the random variable which is the {\em average rank for women} at the end of the process (similarly, $\menrank$ is the average rank for men), and by $\topG{a}$ we denote the event where $g$ is matched with one of her top $a$ preferences.


\paragraph*{Probabilistic notations.} 
We say that an event occurs {\em with high probability (w.h.p.)}, if the probability of non-occurrence approaches zero as $n$ goes to infinity. Throughout the work we use the following multiplicative version of the {\em Chernoff Bound}:
Let $X_1, ..., X_n$ be independent random variables taking values in $\{0, 1\}$, $X = \sum_{i=1}^{n} X_i$ and $\mu = \mathbb{E}[X]$. Then, 
$${ \Pr(X\leq (1-\delta )\mu )\leq e^{-{\frac {\delta ^{2}\mu }{2}}},\qquad 0\leq \delta \leq 1,}$$ 
$${ \Pr(X\geq (1+\delta )\mu )\leq e^{-{\frac {\delta ^{2}\mu }{3}}},\qquad 0\leq \delta \leq 1,}$$ 
$${ \Pr(X\geq (1+\delta )\mu )\leq e^{-{\frac {\delta \mu }{3}}},\qquad 1\leq \delta.}$$ 



\section{Results in the Full Information Model}
\subsection{Main Results}
In this section we show that a strategic behavior of a single agent is expected to dramatically affect the entire outcome of the matching. In particular, the expected rank for women is polylogarithmic in $n$ compared to an expected rank of $\frac{n}{\log(n)}$ under a truthful regime. Similarly, with a probability that goes to $1$ as $n$ grows, the rank of almost all the women is polylogarithmic in $n$.

Our analysis is with respect to the process described in previous section which is equivalent to the Gale--Shapley algorithm when the agents's preferences are drawn uniformly at random and woman $g$ trim her preference list at an optimal point.
We first claim that the strategic behavior does not harm the stableness of the resulted matching.
\begin{lemma}\label{lem: all women matched}
	The process terminates with a stable matching where all the women are matched. 
\end{lemma}
\begin{proof}
	Let $\mathcal{M}$ be the matching resulting from the process. Notice that $\mathcal{M}$ is stable in the reported preferences. For any pair who does not include the strategic woman $g$, the stability for the reported preferences implies stability for the true preferences. For any pair that includes $g$ and a man $b$, either $g$ rejected him and got a better match, or $g$ didn't get an offer from $b$. In the former case, $g$ and $b$ are not a blocking pair due to $g$'s preferences, and in the latter case, $g$ and $b$ are not a blocking pair due to $b$'s preferences. Hence, the matching is stable.
	Since there is a matching where all agents are matched (for example, the women-optimal matching) in any stable matching all agents are matched (see \citet{mcvitie1970stable}).
\end{proof}
Note that since the matching is stable, $g$'s optimal strategy is equivalent to accept only the proposal from her matching in the women-optimal matching. This strategy can be computed in polynomial time using the Gale--Shapley algorithm when men and women switch their roles.

When the preferences are drawn uniformly, we have that with high probability, any player's best stable match is of polylogarithmic order and the expected rank when matched to this match is $\log(n)$ (see \citet{pittel1992likely}). We use these results for estimating the rank of $g$'s match.

\begin{lemma}\label{lem: rank of $g$ is log}
	For any function $\slowfunc(n)$ that goes to infinity and any large enough $n$, we have that:
	\begin{equation}
	\Pr(\topG{7\log^2(n)})> 1-\frac{1}{n},
	\end{equation}
	\begin{equation}
	\lim_{n\rightarrow\infty}\Pr(\topG{\slowfunc(n)\log(n)})= 1.
	\end{equation}
	
\end{lemma}
\begin{proof}
	The first inequality is a corollary of Theorem 6.1 in \cite{pittel1992likely}.
	The second is by Markov's inequality and the fact that the expected rank of the best stable matching is of order $\log(n)$.
\end{proof}

Since the process continues until $g$ get a match, we have that with high probability, the process continues until $g$ gets an offer from one of her top $\log^{1+\epsilon}(n)$ preferences.
We now show that the number of proposals made to $g$, conditioned on this event, is approximately the same as without conditioning on this event.
That is, if $g$ gets matched to one of her top $a$ options we may assume she got order of $\frac{n}{a}$ distinct proposals. Formally,

\begin{lemma} \label{lem: m is almost geometric}
	Let $a>2\log(n)$, and $M$ the total number of distinct (i.e., non-redundant) proposals made to $g$ before the termination of the algorithm. Then, $\Pr(M=m|\topG{a})<\frac{2a}{n}$ for large enough $n$ and any $m\in \mathbb{N}$.
\end{lemma}
\begin{proof}
	We first examine $\Pr(M=m)$ unconditioned on the event $\topG{a}$ and look for the first time $g$ gets a proposal from a man who is her $d$-th option.
	Since the preference list of $g$ is determined online and independently of the proposals she gets, the order of ranks viewed by $g$ is uniform and the probability that $m$ proposals are needed before seeing $b$ is exactly $\frac{1}{n}$.
	Using the union bound, we deduce that $\Pr(M=m)<\frac{a}{n}$ when counting the proposals till $g$ gets a proposal from one of her $a$ most preferred men.
	By Bayes' rule, and conditioning on the fact that we are in the event $\topG{a}$, we get
	$$\Pr(M=m|\topG{a})=\frac{\Pr(M=m,\ \topG{a})}{\Pr(\topG{a})}\leq \frac{\Pr(M=m)}{\Pr(\topG{a})}\leq \frac{2a}{n},$$
	where the last inequality holds by Lemma \ref{lem: rank of $g$ is log}.
\end{proof}

Next, we show that although $g$ is the only strategic woman, almost all women in the game get a similar number of (not necessarily distinct) proposals. 

\begin{lemma} \label{lem: all men get m/4}
	Let $k$ be the total number of proposals made by men. If $k\geq\frac{mn}{2}$ then with a probability of at least $1-n\cdot e^{-\frac{m}{16}}$, all of the women get at least $\frac{m}{4}$ proposals.
\end{lemma}
\begin{proof}
	We bound the probability of an arbitrary woman to get fewer than $\frac{m}{4}$ proposals. Since at each proposal, the woman who gets a proposal is chosen uniformly and independently of previous proposals, we may use the Chernoff bound with parameters $\mu=\frac{m}{2}$ and $\delta=\frac{1}{2}$. Thus, the probability of such an event is less than $e^{-\frac{m}{16}}$.
	By the union bound, the probability of at least one woman not getting enough proposals is bounded by $n\cdot e^{-\frac{m}{16}}$.
\end{proof}

Lemma~\ref{lem: all men get m/4} shows that the number of proposals in the process is distributed in an approximately uniform way among the women. However, this is true only when counting each proposal no matter from which man it came. However, for distinct proposals the claim is not as strong. In the next lemma we show that when conditioning on the event that the process terminates with all men matched, we can get a similar result, albeit a slightly less tight one. This conditioning is reasonable by Lemma \ref{lem: all women matched}.

\begin{lemma}\label{lem: men offer at most log times to each}
	With probability $1-\frac{2}{n^2}$, no men propose to the same woman more than $20\log(n)$ times.
\end{lemma}
\begin{proof}
	Assume that some man, $b$, proposes $r$ proposals in total and let $w$ be some woman.
	We start by showing that with high probability $b$ makes no more than $4n\log(n)$ proposals.  
	Because only full matchings are considered, we know that no man reaches the end of his preference list \footnote{If $b$ is matched to his least preferred option, he reach to the end of the list but only propose once to this woman.}.
	Let $p$ be the probability that $w$ gets no proposal from $b$. We have that $p=(1-\frac{1}{n})^r$ and for $r>4n\log(n)$, we have that $p\leq\frac{1}{n^4}$. 
	By the union bound, after $4n\log(n)$ proposals, the probability that not all women got proposed by $b$ is bounded by $\frac{1}{n^3}$.
	Under the assumption that $r$ is at most $4n\log(n)$ we use the Chernoff bound, with $\mu=4\log(n)$ and $\delta=4$. The probability that more than $20\log(n)$ proposals are made to the same women $w$ is at most
	$$\Pr(B_1)<e^{-4\log(n)} = \frac{1}{n^4}.$$
	
	Summing over all men and women and using the union bound, we get that with probability at least $1-\frac{1}{n^2}-\frac{1}{n^2}$, no man proposes to the same woman more than $20\log(n)$ times.
\end{proof}
\begin{corollary} \label{cor: women get at lest k/log distinct}
	If a woman gets $k$ proposals then with probability at least $1-\frac{2}{n^2}$ she gets at least $\frac{k}{20\log(n)}$ distinct proposals.
\end{corollary}

We next prove formally the intuition that receiving many distinct proposals would lead to a good match.
\begin{lemma} \label{lem: k distinct -> rank an/k}
	Let $a\in\mathbb{R}$ be a constant, $w$ a woman, and $k$ the number of distinct proposals made to $w$. Then, with probability at least $1-e^{-a}$, $w$ is matched with a man who is ranked among her $\frac{an}{k}$ most preferred men.
\end{lemma}
\begin{proof}
	Since $w$'s preferences are independent of the proposals she gets, the probability of her not getting a proposals from any subset of $\frac{an}{k}$ men out of $n$ men is $\left(1-\frac{a}{k}\right)^{k}<e^{-a}$.
\end{proof}
The next lemma shows the connection between the number of proposals in the process and the final rank for the women side.
\begin{lemma} \label{lem: k distinct -> rank an/k FINALE}
	For all $k>0$, if all woman receive at least $k$ distinct proposals, then $\mathbb{E}\womenrank\leq \frac{2n}{k}$ and with probability at least $1-\frac{1}{n}$ it holds that $\womenrank<\frac{6n}{k}+5\log(n)$.
\end{lemma}
\begin{proof}
	Let $X_i$ be the random variable for the rank of woman $i$. We start by estimating $\mathbb{E}X_i$ using the identity $\mathbb{E}X_i = \sum_{j=0}^{n}{\Pr(X_i>j)}$, where $j=\frac{an}{k}$ (and thus $a=\frac{jk}{n}$). 
	\begin{eqnarray*}
	\sum_{j=0}^{n}{\Pr(X_i>j)}&\leq& \sum_{j=0}^{n}{e^{-\frac{jk}{n}}}\\
	&\leq&\sum_{j=0}^{\infty}{\left(e^{-\frac{k}{n}}\right)^j}\\
		&\leq&\frac{1}{1-e^{-\frac{k}{n}}}\\
		&\leq&\frac{1}{\frac{k}{n}-\frac{1}{2}\left(\frac{k}{n}\right)^2}\leq \frac{2n}{k},
	\end{eqnarray*}
	where the first inequality is due to Lemma~\ref{lem: k distinct -> rank an/k} with $a=\frac{jk}{n}$ and the last inequality true for $\frac{k}{n}<1$, which is our case.
	\stam{
	By Lemma~\ref{lem: k distinct -> rank an/k}, 
	$$\sum_{j=0}^{n}{\Pr(X_i>j)}\leq \sum_{j=0}^{n}{e^{-\frac{jk}{n}}} = \sum_{j=0}^{n}{\left(e^{-\frac{k}{n}}\right)^j}\leq \sum_{j=0}^{\infty}{\left(e^{-\frac{k}{n}}\right)^j}\leq \frac{1}{1-e^{-\frac{k}{n}}}, $$
	where the first inequality is due to Lemma~\ref{lem: k distinct -> rank an/k}.
	where the last inequality is due to the convergence of the sum of geometric series. This term is bounded in the following way:
	$$\frac{1}{1-e^{-\frac{k}{n}}}\leq \frac{1}{\frac{k}{n}-\frac{1}{2}\left(\frac{k}{n}\right)^2}\leq \frac{2n}{k}$$
	where the first inequality is due to the approximation $e^{-x}\leq 1-x+\frac{x^2}{2}$ and the second is true for $\frac{k}{n}<1$, which is precisely our case.}
	By the linearity property of the expectation, we get that $\mathbb{E}\womenrank\leq \frac{2n}{k}$ which conclude the first claim of the lemma.
	
	We now show that with high probability, $\womenrank$ is small. We order $X_i$ by size and summing them in segments where in each segment all ranks are roughly the same.
	$$\womenrank = \frac{1}{n}\sum_{a=1}^{k}{\sum_{\frac{(a-1)n}{k}<X_i\leq\frac{an}{k}}}{X_i}\leq\frac{1}{n}\sum_{a=1}^{k}{\sum_{\frac{(a-1)n}{k}<X_i\leq\frac{an}{k}}}{\frac{an}{k}}.$$
	For $a=1$, the number of women with rank less than $\frac{an}{k}$ is of course no more than $n$. Therefore,
	$$\womenrank\leq\frac{1}{n}\left(n\cdot\frac{1\cdot n}{k}+\sum_{a=2}^{k}{\sum_{\frac{(a-1)n}{k}<X_i\leq\frac{an}{k}}}{\frac{an}{k}}\right)$$
	For $a\geq 2$, we show that only a small number of women have a rank which is much larger than $\frac{n}{k}$.
	Using the Chernoff bound with $\mu=ne^{-a}$ and $\delta=4$, we get that the probability of more than $5n\cdot e^{-a}$ of the women getting a worse rank than $\frac{an}{k}$ is at most $e^{-\frac{4ne^{-a}}{3}}$.
	For $a=\log(n)-\log(\log(n))$, this probability is less than $n^{-\frac{4}{3}}$ and we can neglect the event that more than $5ne^{-(\log(n)-\log(\log(n)))}=5\log(n)$ women have rank greater than $\frac{n}{k}\cdot(\log(n)-\log(\log(n)))$ by assuming that the rank of any such women is $n$.
	For $a\in [2,\log(n)]$, by using the same bound we get that the number of women with rank $\frac{an}{k}$ or better is at least $1-5e^{-a}$. 
	If there are more than $1-5e^{-a}$ such women, we count the extra ones as if they are in the next segment (i.e., with rank $\frac{(a+1)n}{k}$ or better), thus only enlarging our estimation of $Y$. Thus, the number of women in each segment is at most $(1-5e^{-a})n-(1-5e^{-(a-1)})n=5n(e^{-(a-1)}-e^{-a})$. Hence ,
	$$\womenrank\leq\frac{1}{n}\left(n\cdot\frac{n}{k} +5\log(n)\cdot n +\sum_{a=2}^{\log(n)-\log(\log(n))}{5n(e^{-(a-1)}-e^{-a})\frac{an}{k}}\right) <\frac{6n}{k}+5\log(n)$$
	with probability larger than $1-\frac{1}{n}$.
\end{proof}

We are now ready to prove our main theorem. Denote by $\event{m}{s}$ the event that $m$ distinct proposals made to $g$ and $s$ (not necessarily distinct) proposals made in total to all women. Recall that $\topG{a}$ is the event that $g$ is matched with one of her top $a$ priorities. We divide our analyses into four disjoint events:
$$B_0^a=\left\{ \event{m}{s}\ |\ m \leq 64\log(n)\right\} \cap \topG{a}$$
$$B_1^a=\left\{ \event{m}{s}\ |\ m > 64\log(n),\ s\leq\frac{mn}{2}\right\} \cap \topG{a}$$
$$B_2^a=\left\{ \event{m}{s}\ |\ m > 64\log(n),\ s>\frac{mn}{2}\right\} \cap \topG{a} $$
and the event that $\topG{a}$ doesn't occur (denoted by $\bar{\topG{a}}$).
The next two Lemmas shows that the events $B_0^a$ and $B_1^a$ have negligible effect on the expected average rank. 
\begin{lemma} \label{lem: B0 small}
	For $a> 2\log(n)$, it holds that $P(B_0^a)<\frac{128a\log(n)}{n}$ and $P(B_0^a)\cdot \mathbb{E}(\womenrank|B_0^a) \leq 128a\log(n)$.
\end{lemma}
\begin{proof}
	For the first event, by Lemma \ref{lem: m is almost geometric}, the probability that exactly $m$ proposals are made to $g$ during the process is at most $\frac{2a}{n}$.
	Using the union bound we get that the probability of $g$ getting less than $64\log(n)$ proposals is bounded by $\frac{128a\log(n)}{n}$.
	Since the maximal possible rank in a full matching is $n$, we have that $P(B_0^a)\cdot E(\womenrank|B_0^a) \leq 128a\log(n)$.
\begin{lemma} \label{lem: B1 small}
	For $a> 2\log(n)$, it holds that $P(B_1^a)<\frac{1}{n^3}$ and $P(B_1^a)\cdot \mathbb{E}(\womenrank|B_1^a) \leq \frac{1}{n^2}$.
\end{lemma}
	In the second event we have that $g$ gets at least $64\log(n)$ proposals.
	We bound the probability that out of the $s$ proposals made in total, at least $m$ are made to $g$ using the Chernoff bound. Set $\mu=\frac{s}{n}\leq \frac{m}{2}$. Then $P(B_1^a)<e^{-\frac{m}{6}}<e^{-64log(n)}=\frac{1}{n^3}$. 
	Since the maximal possible rank in a full matching is $n$, we have that $\Pr(B_1^a)\cdot E(\womenrank|B_1^a) \leq \frac{1}{n^2}$.
\end{proof}
The next lemma implies the first part of Theorem \ref{thm: w.h.p wlog2}, and it used to prove the bound on the average men's rank. The bound for the expected women's rank is be proven separately.
\begin{lemma} \label{lem: w.h.p wlog2}
	Let $\slowfunc(n)$ be any function that goes to infinity as a function of $n$. Then, there exist constants $c,d>0$ such that w.h.p. $\womenrank<c\cdot\slowfunc(n)\log^2(n)$, and any arbitrary women get a match with one of their top $d\cdot\slowfunc(n) \log^2(n)$ options.
\end{lemma}
\begin{proof}
	Let $a =\sqrt{\slowfunc(n)}\log(n)$. By Lemma \ref{lem: rank of $g$ is log}, the event $\topG{a}$ happens with probability that goes to $1$ and by Lemmas~\ref{lem: B0 small} and~\ref{lem: B1 small} the events $B_0^a$ and $B_1^a$ are negligible. Thus, we may assume that event $B_2^a$ happen. 
	By Lemma~\ref{lem: m is almost geometric}, the probability that $m<\frac{n}{\slowfunc(n)\log(n)}$ is at most $\frac{n}{\slowfunc(n)\log(n)}\cdot\frac{2a}{n}=\frac{2}{\sqrt{\slowfunc(n)}}$ which goes to $0$.
	By Lemmas~\ref{lem: all men get m/4} and~\ref{lem: men offer at most log times to each}, with probability at least $1-\frac{2}{n^2}$, each woman gets at least $\frac{m}{80\log(n)}$ distinct proposals.
	We conclude using Lemma~\ref{lem: k distinct -> rank an/k FINALE} that in total, 
	$$\womenrank<\frac{6n}{\frac{m}{80\log(n)}}+5\log(n) \leq 480\slowfunc(n)\log^2(n)+5\log(n),$$
	with probability that goes to $1$ as $n$ grows.
	
	We use the same arguments to show that with high probability, almost all women get a good matching. With high probability, each woman gets at least $\frac{n}{80\slowfunc(n)\log^2(n)}$ distinct proposals (with $\slowfunc(n)$ as before). By Lemma~\ref{lem: k distinct -> rank an/k}, the probability that a woman with that number of proposals getting a rank worse than $80\slowfunc(n)\slowfunc_0(n)\log^2(n)$ is smaller than $e^{-\slowfunc_0(n)}$.
	$\slowfunc_0(n)$ and $\slowfunc(n)$ are chosen such that $\slowfunc_0(n)\slowfunc(n)<\log^\epsilon(n)$, which completes the proof.
\end{proof}


\citet{pittel1992likely} showed that with probability that goes to $1$, any stable matching with an average women's rank of order of $k$ has an average men's rank that is of order $\frac{n}{k}$. Thus, the first part of Lemma \ref{lem: w.h.p wlog2} also bounds the average men's rank.
\begin{corollary}
	For any $\epsilon>0$, w.h.p. the average men's rank, $\menrank$, is of order $\Omega\left(\frac{n}{\log^{2+\epsilon}(n)}\right)$.
\end{corollary}

In the last part of this section, we bound the expected rank of the women.

\begin{lemma}\label{lem:B2}
	There exist $c>0$ such that if $a> 2\log(n)$, then $\mathbb{E}(\womenrank|B_2^a)<ca\log^2(n)$.
\end{lemma}
\begin{proof}
	Conditioned on the event $B_2^a$, we have that the total number proposals $s$ is at least $\frac{mn}{2}$. Since $m>64\log(n)$, by Lemmas \ref{lem: all men get m/4} and \ref{lem: men offer at most log times to each}, we have that all the women get at least $\frac{m}{80\log(n)}$ distinct proposals with probability at least $1-\frac{2}{n^2}-\frac{1}{n^3}$.
	For $a>2\log(n)$, Lemma \ref{lem: m is almost geometric} holds and the probability for receiving exactly $m$ proposals is no more than $\frac{2a}{n}$. Thus, by Lemma \ref{lem: k distinct -> rank an/k FINALE} we conclude that with probability $1-\frac{2}{n^2}-\frac{1}{n^3}$:
	$$\mathbb{E}[\womenrank|B_2^a]\leq \sum_{m=64\log(n)}^{n}{\left(\frac{2a}{n}\right)\cdot\left(\frac{2n}{\frac{m}{80\log(n)}}\right)}\leq c'a\log^2(n),$$
	for some $c'>0$. For the cases this claim does not hold we assume that $\womenrank=n$ and in total there exist $c>0$ as desired.
\end{proof}
Finally, we show that setting $a$ to $7\log^2(n)$ as Lemma~\ref{lem: rank of $g$ is log} suggests, leads to the conclusion that the expected women's rank is $O(\log^4(n))$
\begin{lemma}
	$\mathbb{E}[\womenrank]=O(\log^4(n))$.
\end{lemma}
\begin{proof}
	We write the expectation as the sum of four disjoint terms:
	$$\mathbb{E}[\womenrank] = \Pr(B_0^a)\mathbb{E}[\womenrank|B_0^a]+\Pr(B_1^a)\mathbb{E}[\womenrank|B_1^a]+\Pr(B_2^a)\mathbb{E}[\womenrank|B_2^a]+\Pr(\bar{\topG{a}})\mathbb{E}[\womenrank|\bar{\topG{a}}]$$
	Setting $a=7\log^2(n)$ and using Lemmas~\ref{lem: rank of $g$ is log},~\ref{lem: B0 small},~\ref{lem: B1 small} and~\ref{lem:B2} yields:
	$$\mathbb{E}[X] \leq 128\log^3(n)+\frac{1}{n^2}+1\cdot 7c\log^4(n)+\frac{1}{n}\cdot n =\O(\log^4(n))$$
	which completes the proof of Theorem~\ref{thm: w.h.p wlog2}.
\end{proof}

\subsection{The Set of Stable Husbands}\label{sec: perms}
An alternative point of view about the effect of a strategic woman is given via observation on the set of stable husbands for each woman. A man $m$ is called a \textit{stable husband} of $w$ if there is a stable matching in which they are matched together.

By the optimality of $g$'s strategy, $g$ gets a proposal from her best stable husband $b_0$. Due to the lattice structure of the set of stable matchings, it is known that from $b_0$'s perspective, she is his worst stable wife (e.g., \cite{teo1998geometry}). Thus, we know that at some point in the process $b_0$ proposed to his second-worst woman, $g_0$, and was rejected.
Since all the women are eventually matched, we know that $g_0$ ends up with a proposal from her best stable husband.
In the same manner, this event initiates a series of events that eventually prove some set of women are all guaranteed to be matched to their best stable husband at the end of the process.

Assume that the size of this set is distributed like the size of a cycle in a random permutation. Then, with constant probability, at least a fraction of the women are matched to their best stable husband. The experimental results in part \ref{sec: experimental} support this intuition and show that the expected number of women getting their best stable matching is $\frac{n}{2}$.

Note that $g$ received a proposal not only from her second-best stable husband but also from all of her possible stable husbands. Let $g$'s second-best stable husband be $b_1$. If we could reason that she is his second-worst stable wife, we would have got a new set of women guaranteed to be matched to their second-best stable possible husband. 
By Theorem 2 in \cite{teo1998geometry}, this is true when ordering the stable spouses with repetitions (one for each possible stable matching). Unfortunately, this may not be the case when counting each spouse only once. It will be interesting to see if we could show that w.h.p. this is still true even for the distinct case.

It seems that it would be possible to show that at least $n\cdot2^{-k}$ of the women are expected to be matched with their $k$th stable husband.
Furthermore, it is known \cite{pittel1992likely} that the expected rank of a woman who is matched to her best stable husband is $\O(\log(n))$.
Looking at the process described in the previous section, we observe that each temporarily best proposer to the strategic woman in the rejecting part of the process (i.e., after the men-optimal match is reached) is a stable husband and is randomly located in the proposed-to woman's preference list. Hence, the rank of the woman when she is matched to her $k$th stable husband is expected to be twice as good compared to when she is matched to the $(k+1)$th stable husband. Hence, the expected rank of a woman who is matched to her $k$th stable husband is $2^{k-1}\log(n)$.
Notice that some of the women have few stable husbands (see \cite{pittel2007number}) and it is unclear how to handle this in a rigorous way.
Combining these two observations rigorously hopefully will yield an improved bound of order $\frac{1}{2}\log^2(n)$ for the expected rank of the women.
This observation demonstrates the core convergence since in a strategic environment many of the agents will be left with a single option to be matched with.

\subsection{Comparison to Unbalanced Markets}
An unbalanced market of $n$ men and $n-1$ women can be described as a balanced market of size $2n$ in which one woman rejects all of the proposals made to her. In this case, the men's utility from being matched with her is irrelevant and may be ordered arbitrarily.
In this market, assuming uniform and independent preferences, \citet{ashlagi2017unbalanced} shown that the average rank of the women is of order $O(\log(n))$ even in the women-pessimal matching.
In the context of our paper, a truthful balanced market can be viewed as a market in which this one woman decides not to trim her preference list at all.
In this case, the average women's rank is of order $\O(\frac{n}{\log(n)})$ in the woman-pessimal match.

Adding the results of this paper, we get that a strategic woman can affect this range of ranks in an almost continuous way.
To see this, observe that as long as all the men are matched at the end of the process, all of the lemmas in the previous section hold if we choose a suitable set $\topG{a}$. 
An interesting corollary is that if $g$'s list is of a length which is of order $\sqrt{n}$ and all the other lists are long enough, we get that the women's rank is around $\sqrt{n}$ up to a logarithmic factor. By the hyperbola matching rule, this match captures some of the properties of fairness between the sexes.

Another observation is for the case of lists that are uniformly ordered of different lengths. This scenario can be described as a market with full preferences where each woman decides how she trims her list. In this case, the expected rank is asymptotically determined by the woman with the shortest list. By using different truncations, the women may induce any possible stable matching (see Theorem 4 in \cite{mcvitie1971stable}). It should be noted that when there are many strategic agents there are more profitable actions that a coalition of agents may achieve.

\subsection{Experimental Results}\label{sec: experimental}
In the theoretical part of the work, we only estimated the market behavior as the market size grows to infinity.
For a better understanding of the behavior in typical market sizes, we simulated different markets. We tested our settings for market with $n$ up to $10,000$ in leaps of $20$. For any market
size, the mean of average ranks over $100$ iterations was calculated.
Figure \ref{fig:average rank} shows the mean of the average ranks for men and women, in the truthful and the strategic scenarios. Figure~\ref{fig:average rank closeup} zooms in on the mean of the average ranks for women in the strategic scenario.
\begin{figure}[ht] 
	\centering
	\includegraphics[width=1\textwidth]{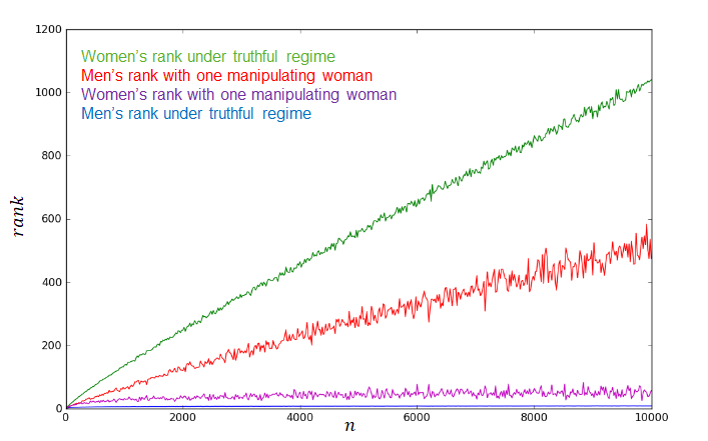}
	\caption{Mean average rank of men and women in both scenarios}
	\label{fig:average rank}
\end{figure}
\begin{figure}[ht]
	\centering
	\includegraphics[width=1\textwidth]{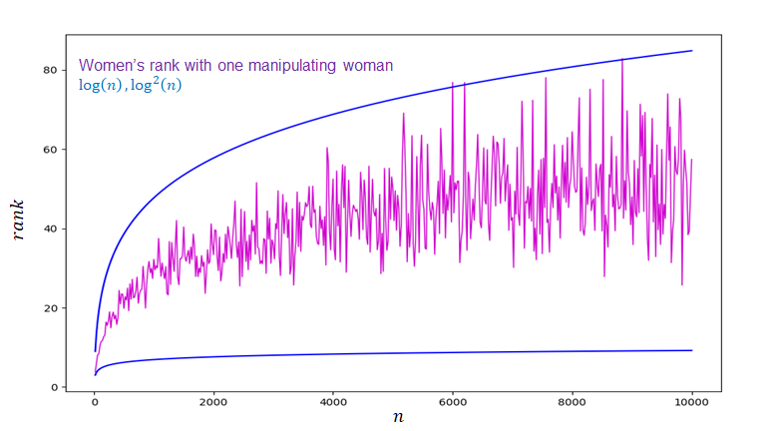}
	\caption{Mean average rank of women with strategic agent compared to $\log(n)$ and $\log^2(n)$}
	\label{fig:average rank closeup}
\end{figure}

For the same settings, we counted the number of women who got their best stable
husband in the strategic scenario. The simulation shows that in expectation half of the women get their best stable husband whenever another woman acts strategically.
The results are shown in Figure \ref{fig:BS}. The simulation supports the intuition described in Section \ref{sec: perms}.
\begin{figure}[ht]
	\centering
	\includegraphics[width=0.8\textwidth]{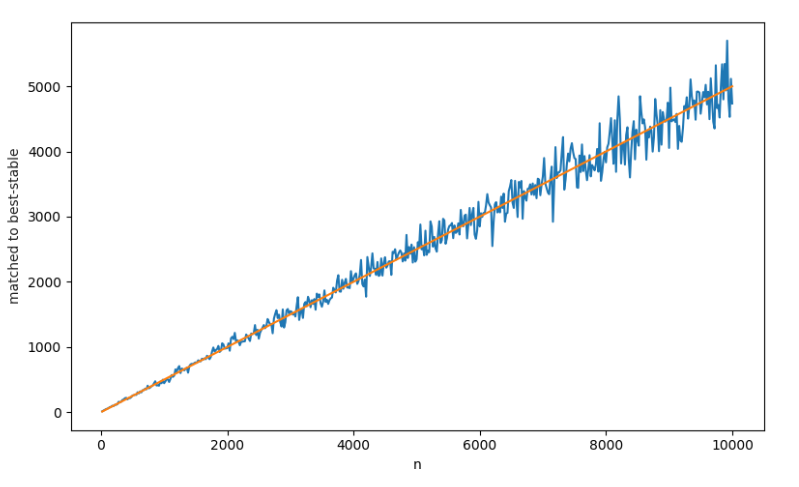}
	\caption{The number of women who are matched with best stable match compared to $\frac{n}{2}$.}
	\label{fig:BS}
\end{figure}

\section{Discussion}
\stam{
	\subsection{The Exact Expected Rank}
	We've shown that the expected women's rank is $\O(\log^4(n))$. A logarithmic factor was added since the number of distinct proposals to an arbitrary woman was bounded by a fraction of order $\O(\frac{1}{\log(n)})$ of the distinct proposals made to $g$. In fact, in most cases, we can show that those two are the same up to a constant multiplicative factor.
	Showing this rigorously will prove that the expected rank is of order $\O(\log^3(n))$. It will also help to conclude that the expected rank is $\Omega(\log^2(n))$. However, additional work is needed in this direction. An alternative direction is to formalize the observations stated in Section \ref{sec: perms}.
}
\subsection{Strategic Behavior in the Bayesian Model}\label{sec: Bayesian information}
In the settings of this work, $g$ is assumed to have full information on the market and thus knows how to choose an optimal strategy. This assumption is later relaxed to having access to an oracle that hints when to terminate a process. 
Most of the Lemmas do not directly use $g$'s knowledge of other preferences. Although it is tempting to deduce that this implies exactly the same results when $g$ knows only the other lists' distributions, it is not the case. For example, assume that $g$ has an extremely high utility for not staying single and hardly distinguishes between all possible mates. In that case, it is likely that $g$'s optimal strategy is just to report truthfully and thus guarantees that she will be matched. But now, $g$ has no effect on the rest of the market and the expected rank for all agents is unchanged.
Note that although $g$'s rank isn't monotone in her truncation and she needs to be careful not to trim her list too much, the average rank of all the other women is indeed monotone in $g$'s decision and they can only benefit from $g$ being too picky.

The optimal strategy of $g$ given she have Bayesian information is discussed in \cite{roth1999truncation,coles2014optimal}.
When the uniform independent assumption holds, a truncation strategy is still optimal for her when she maximizes her expected utility, and the exact point of truncation is determined by her utility function. Furthermore, \citet{coles2014optimal} showed that a reduction in the risk aversion causes the length of $g$'s list after truncation to be negligible compared to the number of agents.
If we assume truncation in $\Theta(\log^2(n))$ we get similar results to those of the full knowledge model: 
either the strategy was successful and the expected rank is $\O(\log^4(n))$ or the strategic woman is unmatched and we are in a $n$-men $(n-1)$-women case. In this case, it should be reasonable to say that the other women get rank $\log(n)$ due to \citet{ashlagi2017unbalanced}. A more rigorous analysis is needed since the fact that $g$'s rank, when matched with her best stable husband, is worse than her truncation location might imply something about the preference lists of the other agents. 
Our result is different from \cite{ashlagi2017unbalanced} in the sense that in our case the \textit{missing woman} not only stays in the game but promises herself the best possible match. Our results can also be seen as a refinement of \cite{ashlagi2017unbalanced} that allows the understanding of the core convergence depending on the degree of selectivity (truncation) that the strategic players are willing to risk.

\subsubsection{Non-Uniform Distributions}
The independent uniform distribution of the preferences is a widely common assumption in theoretical research, even though in real-life matching markets it is not always justified to assume such a distribution. The following distribution (presented in \cite{immorlica2005marriage}) preserves some of the properties that were used while adding a correlation between the preferences. Men still choose the preferences online as in the process before but instead of choosing uniformly between the women, they use an arbitrary distribution $\mathcal{D}^n$ over the set of women, thus making some of the women more popular than others while keeping variety in the men's preferences. It seems the result of our work may be extended to this distribution with minor adjustments assuming $\mathcal{D}^n$ doesn't make any woman extremely popular or unpopular.
The main constraint is that no woman is much more popular than the others; i.e. constraints on $\max_{x,y\in[n]}{\frac{\mathcal{D}^n(x)}{\mathcal{D}^n(y)}}$ or $\mathbb{E}_{x,y\in[n]}{\frac{\mathcal{D}^n(x)}{\mathcal{D}^n(y)}}$ will hopefully be sufficient for us to get similar results to those received in the uniform case.

\subsection{The Effect of Many Strategic Women}
As we have seen, the fact that one agent acts strategically affects in a nontrivial way the quality of the matching from the other women's perspective. This raises some interesting questions about a scenario in which there are several women act strategically.

\paragraph{Full Information Model.}
How many women need to act strategically in order to induce the women-optimal matching? The intuition in Section \ref{sec: perms} suggests that this number should be close to the number of cycles in a random permutation, i.e., an order of $\log(n)$.

\paragraph{Bayesian Information Model.}
Almost all of the women benefit from not reporting their preferences truthfully, given that the other women are truthful.
On the other hand, as soon as at least one woman truncates her list, many of the other women will have already promised a better match and will not gain much from not being truthful (their probability of staying unmatched is insensitive to the truncation of the first woman).
It is interesting to examine the dynamics of the strategies in such environments (with some suitable distributions of the preferences and cardinality over the matchings). 
The existence of equilibrium in truncation strategies under incomplete information is proved in \cite{coles2014optimal}.
Some considerable directions for further research:
\begin{itemize}
	\item In many real-life matchings, the players report their preferences to a centralized mechanism that runs the algorithm for them. Assume that the women report to the mechanism in an order that is known in advance. How would it affect women's reporting strategies?
	\item Assume that men can truncate their lists as well. What would be the equilibrium and how would a centralized mechanism can affect it?
\end{itemize}

	\bibliographystyle{plainnat}
	\bibliography{main}

\end{document}